\documentclass[12pt,reqno]{amsart}
\usepackage{amsmath, amsthm, amsopn, amssymb, microtype}

\usepackage{mathrsfs} 
\usepackage{mathscinet}
\usepackage{bbm}
\usepackage{enumerate, tikz, tikz-cd, etoolbox, intcalc, geometry, caption, subcaption, afterpage, enumitem}
\usepackage[english]{babel}
\usepackage{tkz-graph}
\usetikzlibrary{arrows}
\usepackage{caption}
\title{Non-alternating mean payoff games}
\author{Tom Meyerovitch}

\author{Aidan Young}
\email{mtom@post.bgu.ac.il, youngaid@post.bgu.ac.il}
\address{Ben Gurion University of the Negev.
	Departement of Mathematics.
	Be'er Sheva, 8410501, Israel
}

\usepackage[colorlinks=true,urlcolor=blue,pdfborder={0 0 0}]{hyperref}
\hypersetup{linkcolor=[rgb]{0,0,0.6}}
\hypersetup{citecolor=[rgb]{0,0.6,0}}
\usepackage[nameinlink]{cleveref}
\crefname{theorem}{Theorem}{Theorems}
\crefname{thm}{Theorem}{Theorems}
\crefname{mainthm}{Theorem}{Theorems}
\crefname{lemma}{Lemma}{Lemmas}
\crefname{lem}{Lemma}{Lemmas}
\crefname{remark}{Remark}{Remarks}
\crefname{prop}{Proposition}{Propositions}
\crefname{defn}{Definition}{Definitions}
\crefname{cor}{Corollary}{Corollaries}
\crefname{section}{Section}{Sections}
\crefname{figure}{Figure}{Figures}
\crefname{quest}{Question}{Questions}
\crefname{notation}{Notation}{Notations}
\crefname{conv}{Convention}{Conventions}
\crefname{example}{Example}{Examples}

\begin{document}

\newtheorem{thm}{Theorem}[section]
\newtheorem{lemma}[thm]{Lemma}
\newtheorem{prop}[thm]{Proposition}
\newtheorem{cor}[thm]{Corollary}
\newtheorem{claim}[thm]{Claim}
\newtheorem{quest}[thm]{Question}
\newtheorem{fact}[thm]{Fact}
\theoremstyle{definition}
\newtheorem{definition}[thm]{Definition}
\newtheorem{example}[thm]{Example}
\newtheorem{remark}[thm]{Remark}
\newtheorem{conj}[thm]{Conjecture}
\newtheorem{notation}[thm]{Notation}
\newtheorem{conv}[thm]{Convention}

\newcommand{\N}{\mathbb{N}}
\newcommand{\Z}{\mathbb{Z}}
\newcommand{\Q}{\mathbb{Q}}
\newcommand{\R}{\mathbb{R}}
\newcommand{\C}{\mathbf{C}}
\newcommand{\cL}{\mathcal{L}}
\newcommand{{\cP}}{\mathcal{P}}
\newcommand{\Pna}{\mathcal{P}^\Omega}
\newcommand{\oP}{\overline{P}}
\newcommand{\cS}{\mathcal{S}}
\newcommand{\cT}{\mathcal{T}}  
\newcommand{\cV}[1]{\mathcal{V}_{#1}}
\newcommand{\cE}[1]{\mathcal{E}_{#1}}
\newcommand{\Aval}{\mathbb{V}_\mathit{alt}}
\newcommand{\NAval}{\mathbb{V}_\mathit{non-alt}}
\newcommand{\per}{\operatorname{per}}
\newcommand{\cW}{\mathcal{W}}

\newcommand{\SV}{\operatorname{SV}}
\newcommand{\Var}{\operatorname{Var}}
\newcommand{\MaxEavg}{\operatorname{MaxAvg}}
\newcommand{\Prob}{\operatorname{Prob}}

\begin{abstract}
We present and study a variant of the mean payoff games introduced by A. Ehrenfeucht and J. Mycielski. In this version, the second player makes an infinite sequence of moves only after the first player's sequence of moves has been decided and revealed. Such games occur in the computation of the covering radius of constrained systems, a quantity of interest in coding theory.
\end{abstract}

\maketitle

\section{Introduction}

In this note, we discuss a deterministic, multi-round, zero-sum, two-player game, which is a variation on the mean payoff games introduced by A. Ehrenfeucht and J. Mycielski in \cite{EhrenfeuchtMycielski}. We call these games ``non-alternating mean payoff games" (or ``non-alternating games" for short), where the first player (Alice) selects an infinite sequence of moves, and the second player (Bob) chooses an infinite sequence of moves with knowledge of Alice's choices. This stands in contrast to the mean payoff games as typically understood, where the two players Alice and Bob make decisions in alternating turns. Although the non-alternating variant might seem less natural from a game-theoretic perspective, it naturally occurs in the computation of the covering radius of constrained systems, a quantity of interest in coding theory (c.f. \cite{MR4751633}), which we describe below. The non-alternating version of this game also occurs as a particular instance of more general framework called ``adversarial ergodic optimization" (c.f. \cite{young2025adversarialergodicoptimization}).

We now describe these games. Since we will mostly be interested in ``long" walks in the graph, our notation follows conventions from the field of symbolic dynamics essentially as in, for example, \cite{LindMarcus}.

A \emph{graph} $G$ consists of a finite set of vertices $\cV{G}$ and a finite set of (directed) edges $\cE{G}$. Every edge $e \in \cE{G}$ has a \emph{source} $s(e) \in \cV{G}$ and a \emph{target} $t(e) \in \cV{G}$. A \emph{walk} of length $n$ in the graph $G$ is a tuple $(e_1,\ldots,e_n) \in \cE{G}$ such that $t(e_k)=s(e_{k+1})$ for every $1\le k <n$. We denote the set of length $n$ walks in $G$ by $\cW_n(G)$, and by $\cW_n(G,e)$ all the walks of length $n$ that start at the vertex $t(e)$.

In both the alternating and the non-alternating games, two graphs $G$ and $H$ are given to players Alice and Bob, respectively, along with a score function $P:\cE{G} \times \cE{H} \to \mathbb{R}$, and a pair of initial vertices $v_0 \in \cV{G}$ and $u_0 \in \cV{H}$. In both the alternating and non-alternating versions of the $n$-round game, Alice eventually chooses a walk $(e_1,\ldots,e_n) \in \cW_n(G)$ of $n$ edges in the graphs $G$ and $H$ starting at $v_0$, and Bob eventually chooses a walk $(f_1, \ldots,f_n) \in \cW_n(H)$ starting at $u_0$.
Bob then pays Alice
\[\sum_{j = 1}^n  P(e_j, f_j).
\]

While both the alternating and non-alternating versions of this game involve each player constructing a walk, the versions are distinguished by how these walks are constructed. In the $n$-round alternating game (a finitary version of the mean payoff game), Alice and Bob alternately choose the edges in the following manner: In the first round, Alice chooses an edge $e_1 \in \cE{G}$ such that $s(e_1)=v_0$, after which Bob picks an edge $f_1 \in \cE{H}$ such that $s(f_1) = u_0$. In the second round, Alice chooses and edge $e_2 \in \cE{G}$ such that $s(e_2) = t(e_1)$, after which Bob picks an edge $f_2 \in \cE{H}$ such that $s(f_2) = t(f_1)$. This continues until each player has made a walk of length $n$. Importantly, at each point in the game, each player has perfect information about all moves they and their opponent has made so far, and in particular, the current location of both players at every moment is known to both players.

On the other hand, in the $n$-round non-alternating game, Alice first chooses a walk $(e_1,\ldots,e_n)$ in graph $G$ with $s(e_1)=v_0$, and then Bob chooses a walk $(f_1,\ldots,f_n)$ in the graph $H$ with $s(f_1)=u_0$. We stress that in the non-alternating game, Bob's decision is allowed to depend on the entire walk that Alice has chosen.

Of course, if at some point before the game ends, either Alice or Bob arrive at a ``sink," i.e. a vertex with no outgoing edges, then the rules are undefined, so we will always assume that the graphs $G$ and $H$ have no sinks.

The alternating and non-alternating games also have infinite-round versions, which proceed similarly, except that the walks $(e_1,\ldots,e_n,\ldots)$ and $(f_1,\ldots,f_n,\ldots)$ are now infinite sequences. The infinite series $\sum_{n=1}^\infty P(e_n,f_n)$ does not typically converge (except for rare and degenerate cases where $P(e_n,f_n) = 0$ for all large $n$), so we instead define the total payoff in terms of an ``asymptotic time average:"
\begin{equation}\label{eq:lim_star_av}
\liminf_{n \to \infty} \frac{1}{n}\sum_{j = 1}^n  P(e_j, f_j) .
\end{equation}

As mentioned, certain types of non-alternating games occur in the computation of the covering radius of a constrained system, a problem of interest in coding theory \cite{MR4751633}. We briefly explain this problem now. A \emph{constraint system} (over a binary alphabet), sometimes also known as a \emph{shift of finite type}, is a collection of words $\mathcal{C} \subseteq \{0,1\}^* := \bigcup_{n=0}^\infty \{0,1\}^n$ induced by a set $F \subset \{0,1\}^k$ of binary strings for some $k \in \N$, where $\mathcal{C}$ is the set of all binary words that contain none of the ``forbidden patterns" in $F$ as subwords. That is, $\mathcal{C}:=\bigcup_{n=0}^\infty \mathcal{C}_n$, where
\[
\mathcal{C}_n := \left\{ (b_1,\ldots,b_n) \in\{0,1\}^n ~:~ (b_j,\ldots,b_{j+k-1}) \not \in F \mbox{ for every } 1\le j < n-k+1 \right\}.
\]
The \emph{hamming distance} between two $n$-tuples $u,w \in \{0,1\}^n$ of equal length is the value \[d_H(u, v) : = \left| j \in [1, n] : u(j) \neq v(j) \right|=\sum_{j=1}^n \delta(u_j,v_j).\]
The \emph{covering radius} $R(\mathcal{C}_n)$ of $\mathcal{C}_n$ is the minimal $r>0$ such that the union of radius $r$-Hamming balls centered at each element of $\mathcal{C}_n$ cover all the words $\{0,1\}^n$. This quantity is exactly equal to the smallest number $r \in \N$ such that any word $u \in \{0,1\}^n$ can be altered to some ``constrained" word $w \in \mathcal{C}_n$ by flipping at most $r$ bits. That is,
\[
R(\mathcal{C}_n) := \max_{ u \in \{0,1\}^n} \min_{w \in F_n(L)} d_H(u,v)=\max_{ u \in \{0,1\}^n} \min_{w \in F_n(L)}\sum_{j=1}^n\delta(u_j,v_j).
\]
The \emph{asymptotic covering radius} of the code $\mathcal{C}$ is defined as
\[
R(\mathcal{C}):= \liminf_{n \to \infty}\frac{1}{n}R(\mathcal{C}_n).
\]

These covering radii and the asymptotic covering radii are quantities of interest in coding theory, in particular for analyzing and optimizing certain coding schemes, as in \cite{MR4751633}.
Given a set of forbidden words, one can directly construct graphs $G$ and $H$ and a $\{0,1\}$-valued score function $P$ so that the covering radius $R(\mathcal{C}_n)$ is equal to the optimal payoff that Alice can secure in the non-alternating $n$-round game. Although the asymptotic covering radius has been explicitly computed for several codes of interest in \cite{MR4751633}, it is not yet known if there is a general procedure for computing an exact value of the asymptotic covering radius. 

Expressions of the form \eqref{eq:lim_star_av} are reminiscent of ergodic averages. The study of ``extreme values" of ergodic averages and how they are obtained is the subject of a subfield of dynamical systems called \emph{ergodic optimization}. See \cite{MR4000508} for a rather comprehensive survey of this topic. In the degenerate case where either Bob or Alice do not have any choices (i.e. if one of the graphs $G, H$ is a directed cycle, so there is only one outgoing edge at each vertex), the alternating and non-alternating versions of the game coincide, and the problem coincides with the ergodic optimization of a locally constant function on a subshift of finite type, a well-studied and understood problem (c.f. \cite{MR2176962}). The optimal value in the infinite-round non-alternating game is also a special case of the ``adversarial ergodic optimization" framework studied in \cite{young2025adversarialergodicoptimization}.

In this note, we formulate the setup of our non-alternating mean payoff games. We also state and prove some basic results about what happens if both players play ``optimally," in a standard game-theoretic sense. The main purpose of this note is to study the non-alternating mean payoff game and compare it to the more studied alternating mean payoff game. We show that if the underlying graphs are irreducible, then both players have optimal strategies in the non-alternating game.

This note is organized as follows:

In \Cref{section:main_results_alternating game}, we introduce some general definitions that we will be using throughout the note. In particular, we provide a rigorous description of the alternating and non-alternating games we consider in this note, their associated strategies, and their payoffs. The reader who is already comfortable with the basic terminology of game theory will find much of this familiar. In \Cref{section:non_alternating_game}, we consider the non-alternating game, showing that when the underlying graphs are irreducible, the Nash equilibrium payoff of the infinite-round non-alternating game is equal to the normalized limit of the Nash equilibrium payoffs of $n$-round non-alternating games. Moreover, each player has an ``optimal strategy" for this infinite-round non-alternating game. Finally, in \Cref{section:open_questions}, we present several open questions about the non-alternating games, as well as some relevant examples. These examples also help shed light on the relation between non-alternating and alternating games.

\noindent\textbf{Acknowledgment:} This research was supported by Israel Science Foundation grant no. 985/23.

\section{Defining the non-alternating game}\label{section:main_results_alternating game}

Formally, in the $n$-round alternating game, a \emph{pure strategy for Alice} is a measurable function from a pair of initial vertices $v_0,u_0 \in \cV{G} \times \cV{H}$ and a length $n$-walk $(f_1,\ldots,f_n)$ in $H$ that begins at $u_0$ into the set of length $n$ walks in $G$ that start at $v_0$, such that the $k$th edge $e_k$ of the output is measurable with respect to $v_0,u_0$ and $f_1,\ldots,f_{k-1}$. These measurability conditions express the idea that Alice only knows Bob's choices of the first $k-1$ edges in the $k$th round when choosing $e_k$, and as such, $e_k$ is a function of $(v_0, u_0, (f_1, \ldots, f_{k - 1})$. Similarly, a \emph{pure strategy for Bob} is a measurable function from a pair of initial vertices $(v_0,u_0) \in \cV{G} \times \cV{H}$ and a walk $(e_1,\ldots,e_n)$ in $G$ of length $n$ that begins at $v_0$ into the set of length $n$ walks in $H$ that start at $u_0$, such that the $k$'th edges $f_k$ is measurable with respect to $v_0,u_0$ and $(e_1,\ldots,e_{k})$. Note that the $k$th term of Bob's output is allowed to depend on the first $k$ edges of a walk in $G$, which corresponds to the fact that Bob is choosing his $k$th edge only after Alice has chosen her $k$th edge.

There is also the notion of a ``mixed strategy" or ``random strategy," which is a probability measure on the space of pure (or ``deterministic") strategies. If Alice and Bob's strategies are mixed, the payoff usually refers to the expected payoff. Although mixed strategies are very important in game theory, we can restrict our attention to pure strategies for most of our discussion. As such, for the remainder of this article, every strategy we consider will be pure unless otherwise stated.

In the $n$-round non-alternating game, a strategy for Alice is just a function from a pair of initial vertices $(v_0,u_0) \in \cV{G} \times \cV{H}$ into the set of length $n$ walks in $G$ that start at $v_0$, whereas a strategy for Bob is a function from a pair of initial vertices $(v_0,u_0) \in \cV{G} \times \cV{H}$ and a walk of length $n$ in $G$ into the set of walks in $H$ of length $n$ that start at $u_0$. If Alice and Bob's strategies are $\cS_A$ and $\cS_B$, respectively, then 
the payoff of the $n$-round game started at $(v_0,u_0) \in \cV{G}\times \cV{H}$ is the value
\begin{equation*}
    \oP_n(\cS_A,\cS_B\mid v_0,u_0):= \sum_{j = 1}^n  P(e_j, f_j) ,
\end{equation*}
where $(e_1,e_2,\ldots,e_n) \in \cW_n(G)$ and $(f_1,f_2,\ldots,f_n) \in \cW_n(H)$ are the walks chosen by Alice and Bob respectively according to the strategies $\cS_A$ and $\cS_B$ with initial vertices $(v_0,u_0)$, meaning that
\[
(e_1,\ldots,e_n) = \cS_A\left(v_0,u_0,(f_1,\ldots,f_n)\right)
\]
and
\[
(f_1,\ldots,f_n) = \cS_B \left(v_0,u_0,(e_1,\ldots,e_n)\right).
\]

Strategies for the alternating and non-alternating infinite-round games are defined similarly, except that now the inputs and outputs are infinite walks.
If Alice and Bob's strategies for the infinite-round game are $\cS_A$ and $\cS_B$, respectively, then the payoff of the $n$-round game started at $(v_0,u_0) \in \cV{G}\times \cV{H}$ is the value
\begin{equation}\label{eq:total_payoff_infinite_round}
    \oP(\cS_A,\cS_B\mid v_0,u_0):= \liminf_{n \to \infty} \frac{1}{n}\sum_{j = 1}^n  P(e_j, f_j),
\end{equation}
where $(e_1,\ldots,e_n,\ldots) \in \cW(G)$ and $(f_1,\ldots,f_n,\ldots) \in \cW(H)$ are the walks chosen by Alice and Bob respectively according to the strategies $\cS_A$ and $\cS_B$ with initial vertices $(v_0,u_0) \in \cV{G} \times \cV{H}$.

These infinite-round alternating games form a special case of the \emph{mean payoff games} introduced by A. Ehrenfeucht and J. Mycierski in \cite{EhrenfeuchtMycielski}. As they showed, these games are positionally determined (c.f. \cite[Theorem 2]{EhrenfeuchtMycielski}).

In this note, we will focus primarily on the non-alternating version of this game.

A \emph{Nash equilibrium} is a pair of strategies $(\cS_A^*,\cS_B^*)$ for Alice and Bob, respectively, so that for every strategies $\cS_A,\cS_B$ for players Alice and Bob, respectively, and every pair of initial vertices $(v_0,u_0) \in \cV{G}\times \cV{H}$, we have
\[
\oP_n(\cS_A,\cS_B^* \mid v_0,u_0) \le \oP_n(\cS_A^*,\cS_B^* \mid v_0,u_0) 
\]
and
\[
\oP_n(\cS_A^*,\cS_B \mid v_0,u_0) \ge \oP_n(\cS_A^*,\cS_B^* \mid v_0,u_0) .
\]

The existence of strategies that are a Nash equilibrium for both versions of the $n$-round game (i.e. alternating and non-alternating) is rather trivial and elementary, because there are no simultaneous decisions and the game is finite: It is well-known that every perfect information finite game admits a (pure) Nash equilibrium (c.f. \cite[Proposition 99.2]{OsborneRubinstein}). Although a Nash equilibrium need not be unique, it is easy to see that the payoff 
of the $n$-round at any Nash equilibrium is the same at any Nash equilibrium.

The value of the zero-sum two-player game at a Nash equilibrium is the highest payoff that Alice can secure, regardless of Bob's strategy. It is also equal to the lowest payoff that Bob can secure, regardless of Alice's strategy. We refer to this number as the \emph{Nash equilibrium payoff} of the game. We denote the Nash equilibrium payoff for the $n$-round alternating game with initial vertices $(v_0,u_0) \in \cV{G} \times \cV{H}$ by
\begin{equation*}
    \Aval(n \mid v_0,u_0)  := \oP_n(\cS_A^*,\cS_B^* \mid v_0,u_0),
\end{equation*}
where $(\cS_A^*,\cS_B^*)$ is a Nash equilibrium for the $n$-round alternating game.

For various notational purposes, it will sometimes be convenient to define the Nash equilibrium payoff with initial \emph{edges}, according to the terminal vertices of the edges. That is, for $n \in \N$ and $(e_0,f_0) \in \cE{G}\times \cE{H}$ we define:
\begin{equation*}
    \Aval( n \mid e_0,f_0):= \Aval(n \mid t(e_0),t(f_0)).
\end{equation*}
Similarly, we denote the Nash equilibrium payoff for the $n$-round non-alternating game with initial vertices $(v_0,u_0) \in \cV{G} \times \cV{H}$ by
\begin{equation*}
    \NAval(n \mid v_0,u_0)  := \oP_n(\cS_A^*,\cS_B^* \mid v_0,u_0),
\end{equation*}
where $(\cS_A^*,\cS_B^*)$ is a Nash equilibrium for the $n$-round non-alternating game.

The payoff of the non-alternating $n$-round game at a Nash equilibrium with initial vertices $(v_0,u_0) \in \cV{G} \times \cV{H}$ is given by the following expression:
\begin{equation*}
    \NAval(n \mid e_0,f_0) = \max_{(e_1,\ldots,e_n) \in \cW_n(G, e_0)} \min_{(f_1,\ldots,f_n) \in \cW_n(H, f_0)}\sum_{j=1}^n P(e_j,f_j) .
\end{equation*}
Existence of a Nash Equilibrium for the infinite-round games requires some argument. In the case of the alternating games, this follows from the game being positionally determined (c.f. \cite[Theorem 2]{EhrenfeuchtMycielski}). In the non-alternating game, a Nash equilibrium will always exist assuming that the graphs $G, H$ are irreducible, as we will show in \Cref{thm:non_alt_nash_lim}.

We denote the Nash equilibrium payoff for the infinite-round alternating and non-alternating games with initial vertices $(v_0,u_0) \in \cV{G} \times \cV{H}$ by $\Aval(v_0, u_0)$ and $\NAval(v_0,u_0)$, respectively. By definition, once we establish the existence of a Nash equilibrium, the payoff at Nash equilibrium is equal to
\[
\sup_{\cS_A \in \mathbb{S}_A}\inf_{\cS_B\in \mathbb{S}_B} \oP(\cS_A,\cS_B \mid v_0,u_0)=\inf_{\cS_B\in \mathbb{S}_B}\sup_{\cS_A \in \mathbb{S}_A}\oP(\cS_A,\cS_B \mid v_0,u_0),
\]
where $\mathbb{S}_A$ and $\mathbb{S}_B$ are the spaces of strategies for Alice and Bob, respectively. The same definition applies in the alternating and the non-alternating case, except that the spaces of strategies $\mathbb{S}_A$ and $\mathbb{S}_B$ have different meaning.

The following result is implicit in \cite{EhrenfeuchtMycielski}, where the authors use the clever technique of constructing a corresponding \emph{finite} deterministic game that is ``equivalent" to the mean payoff game, thus concluding that the game is positionally determined.

\begin{thm}\label{thm:Nash_limit_alternating}
    For any initial edges $(e_0,f_0) \in \cE{G} \times \cE{H}$, the normalized Nash-equilibrium payoff of the $n$-round alternating game converges to the Nash equilibrium of the infinite-round alternating game. That is:
     \[
   \Aval(e_0,f_0) = \lim_{n \to \infty}\frac{1}{n}\Aval( n\mid e_0,f_0).
   \]
   Moreover, there exist positional strategies $\psi_A : \cE{G} \times \cE{H} \to \cE{G}, \psi_B : \cE{G} \times \cE{H} \to \cE{H}$ for Alice and Bob, respectively, which
   are optimal up to a constant for any $n$-round game in the following sense: There exists $M>0$
   so that for every $n \in \N$, for any strategies $\cS_A$ and $\cS_B$ for Alice and Bob, respectively, and any pair of initial edges $(e_0,f_0) \in \cE{G} \times \cE{H}$, we have
   \begin{align*}
   \oP_n(\psi_A,\cS_B \mid e_0,f_0)    & \ge \Aval(n \mid e_o,f_0) - M , \\
   \oP_n(\cS_A,\psi_B \mid e_0,f_0)    & \le \Aval(n \mid e_o,f_0) + M.
   \end{align*}
   In particular, this pair $(\psi_A, \psi_B)$ is a Nash equilibrium for an infinite-round game.
\end{thm}

\section{Nash equilibrium in the non-alternating game}\label{section:non_alternating_game}

Our main result of this section is as follows.

\begin{thm}\label{thm:non_alt_nash_lim}
Let $G$ and $H$ be irreducible graphs, and let $P: \cE{G} \times \cE{H} \to \R$ be a score function. Then there exists a Nash equilibrium $(\cS_A^*,\cS_B^*)$ for the infinite-round non-alternating game, and the payoff at Nash equilibrium with any starting edges $(e,f) \in \cE{G} \times \cE{H}$ is equal to the limit of the normalized payoff at Nash equilibrium for the $n$-round non-alternating game:
    \begin{equation}\label{eq:NAVAl_lim}
        \NAval( e,f) = \lim_{n \to \infty}\frac{1}{n} \NAval(n \mid e,f).
    \end{equation}
    If $(e_1,\ldots,e_n,\ldots)$ and $(f_1,\ldots,f_n,\ldots)$ are the walks produced by the strategies $(\cS_A^*,\cS_B^*)$ and initial edges $(e_0, f_0) \in \cE{G} \times \cE{H}$, then the limit 
    \[
    \lim_{n \to \infty}\frac{1}{n} \sum_{j=1}^n P(e_j,f_j) 
    \]
    exists, and it is thus equal to $\NAval( e_0,f_0)$.
    Moreover, $\NAval( e_0,f_0)$ only depends on the irreducible component of $(e_0,f_0)$ in $G \times H$.   
\end{thm}

\begin{remark}
If we drop the assumption that $G, H$ are irreducible, then \Cref{thm:non_alt_nash_lim} will not necessarily hold, as we see in \Cref{ex:non_alt_reducible_graphs_infinite_limit_counterexample}. Contrast this with the alternating game, c.f. \Cref{thm:Nash_limit_alternating}.
\end{remark}

\begin{lemma}\label{lem:non_alt_lim_exists}
Let $G$ and $H$ be irreducible graphs, and $P: \cE{G} \times \cE{H} \to \R$ a score function. Then the limit $\lim_{n \to \infty} \frac{1}{n} \NAval(n \mid e, f)$ exists.
\end{lemma}

\begin{proof}
    Choose some $(e_0,f_0) \in \cE{G} \times \cE{H}$. and let
    \[
    r_n := \NAval( n \mid e_0,f_0).
    \]
    We will show that the sequence $(r_n)_{n=1}^\infty$ has the following weak form of subadditivity: there exists a constant $C>0$ such that
    \begin{equation}\label{eq:r_n_almost_subaddtive}
    r_{n+m} \le r_n +r_m +C ~\forall n,m \in \N.    
    \end{equation}
    By a generalization of Fekete's subadditivity lemma due to De Brujin and Erd\"os (c.f. \cite{DeBruijinErdos}), it will then follow that the limit $\lim_{n \to \infty}\frac{r_n}{n}$ exists.

    Let us prove \eqref{eq:r_n_almost_subaddtive}.
    We can show that there exists a constant $C>0$ such that for any irreducible component $K$ of $G \times H$, and any pair of initial edges $(e_0,f_0),(\tilde e_0,\tilde f_0) \in \cE{K}$,
    we have
    \begin{equation}\label{eq:NAVAL_diff_bounded}
    \left| \NAval(n \mid e_0,f_0) -\NAval \left( n \mid \tilde e_0, \tilde f_0 \right)\right| \leq C.    
    \end{equation}
    Given a walk $(e_1,\ldots,e_n) \in \cW_n(G)$ and $f_0 \in \cE{H}$, let 
    \begin{equation*}
        W(e_1,\ldots,e_n; f_0) = \min\left\{\sum_{k=1}^{n}P(e_k,f_k)~:~ ~(f_k)_{k=1}^{n}  \in \cW_{n}(H),~  f_1 \in F(f_0)\right\} .
    \end{equation*}
    Choose any $n,m \in \N$.
    Let $(e_1,\ldots,e_n,e_{n+1},\ldots,e_{n+m}) \in \cW_{n+m}(G)$ be a walk with $s(e_1)=t(e_0)$ such that 
    \begin{align*}
      \NAval(n+m  \mid e_0,f_0)= W(e_1,\ldots,e_{n+m} ; f_0).
    \end{align*}
    
    Then
    \begin{align*}
    \NAval(n+m\mid e_0,f_0) & \le \NAval(n \mid e_0,f_0) + \min_{(e_n, f_n) \in \cE{K}} W(e_{n + 1}, \ldots, e_{n + m}; f_n) \\
        & \leq \NAval(n \mid e_0,f_0) + \min_{(e_n, f_n) \in \cE{K}} \NAval \left( m \mid e_n, f_n \right) \\
     \textrm{[\eqref{eq:NAVAL_diff_bounded}]}& \leq \NAval (n \mid e_0, f_0) + \NAval(m \mid e_0, f_0) + C .
    \end{align*}
    This proves \eqref{eq:r_n_almost_subaddtive}, thus establishing the existence of the limit $\lim_{n \to \infty} \frac{1}{n} \NAval(n \mid e_0, f_0)$.
\end{proof}

\begin{proof}[Proof of \Cref{thm:non_alt_nash_lim}]
    We now describe a pair of strategies $\left(\cS_A^*,\cS_B^* \right)$ for Alice and Bob, starting with Alice's strategy $\cS_A^*$. Before describing Alice's strategy in detail, we provide a brief sketch: Alice takes a sequence of natural numbers $a_1 < b_1 < a_2 < b_2 < \cdots$ such that $b_k - a_k \to \infty$ and $a_{k + 1} - b_k$ is bounded. Alice then chooses her moves $\left( e_{a_k}, \ldots, e_{b_k} \right)$ as though she were playing the game with $b_k - a_k$ turns starting at $(e_0, f_0)$. This accounts for Alice's moves on the turns $j \in [a_k, b_k]$, and yields a sequence of walks of length $b_k - a_k$, although there is still space between these segments, so for the turns $j \in [b_k + 1, a_{k + 1} - 1]$, Alice chooses edges which connect those segments into a true walk $(e_j)_{j = 1}^\infty$. Now we fill in the details of this sketch.
    
    Fix $(e_0, f_0) \in \cE{G} \times \cE{H}$, and set $\lambda = \lim_{n \to \infty} \frac{1}{n} \NAval(n \mid e_0, f_0)$, which we know to exist by \Cref{lem:non_alt_lim_exists}. Set $p = \operatorname{lcm} \left( \per(G), \per(H) \right)$, where $\per(G)$ denotes the \emph{period} of an irreducible graph $G$ as defined in, for example, \cite[\S 4.5]{LindMarcus}. Then there exists a natural number $D \in \N$ with the following property: For any $\left( e', f'\right), \left(e'',f''\right) \in \cE{G \times H}$, if there exists a walk in $G \times H$ of length $n$ from $\left(e', f' \right)$ to $\left(e'', f'' \right)$, where $n$ is divisible by $p$, then there exists a walk from $\left( e', f'\right)$ to $\left( e'', f'' \right)$ of length \emph{exactly} $p D$.

    Consider the sequence of natural numbers $0 = a_1 < b_1 < a_2 < b_2 < \cdots$ defined as follows: for every $k \in \N$, set $b_k = a_k + p k$, and $a_{k + 1} = b_k + p D$. We then define a walk $\left( e_j \right)_{j = 0}^{\infty}$ as follows: set $e_{a_k} = e_0$ for all $k \in \N$, and choose $\left( e_{a_k + 1}, \ldots, e_{b_k}, e_{b_k + 1}, \ldots, e_{a_{k + 1} - 1} \right)$ such that
    \begin{itemize}
        \item $W \left( e_{a_k + 1}, \ldots, e_{b_k} ; f_0 \right) = \NAval\left( b_k - a_k \mid e_0, f_0 \right)$, and
        \item $\left( e_j \right)_{j = 0}^{\infty}$ is a walk in $G$.
    \end{itemize}
    Note that because the $a_k, b_k$ are all divisible by $p$, and $a_{k + 1} - b_k = p D$, we can always ``fill in" the edges $\left( e_{b_k + 1}, \ldots, e_{a_{k + 1} - 1} \right)$ in order to make $\left( e_{b_k}, e_{b_k + 1}, e_{a_{k + 1} - 1}, e_0 \right)$ a valid walk in $G$.

    For $k_0 \in \N$, choose a walk $\left(f_j^*\right)_{j = 0}^{b_{k_0}}$ in $H$ such that $f_0 = f_0^*$ and
    \begin{align*}
    W \left( e_1, \ldots, e_{b_{k_0}}; f_0 \right) & = \sum_{j = 1}^{b_{k_0}} P \left( e_j, f_j^* \right) .
    \end{align*}
    Then
    \begin{align*}
    W \left( e_1, \ldots, e_{b_{k_0}}; f_0 \right)
    \geq    & (k_0 - 1) p D \|P\| + \sum_{k = 1}^{k_0} \sum_{j = a_{k} + 1}^{b_k} P \left(e_j, f_j^*\right) \\
     \geq   & (k_0 - 1) p D \|P\| + \sum_{k = 1}^{k_0} W \left( e_{a_k + 1}, \ldots, e_{b_k} ; f_{a_k}^* \right) .
    \end{align*}
    Set
    \begin{align*}
    N_k \left( e_{a_k + 1}, \ldots, e_{b_k} ; f \right) & : = \min_{\left( f_{j} \right)_{j = a_k + 1}^{b_k} \in \cW_{b_k - a_k} (f) } \sum_{j = a_k + p D}^{b_k} P(e_j, f_j) \\
        & = \min_{f_{a_k + p D} \in F^{p d} (f)} \left( \min_{(f_j)_{j = a_k + p D + 1}^{b_k} \in \cW_{b_k - a_k - p D} \left( f_{a_k + p D} \right)} \sum_{j = a_k + p D + 1}^{b_k} P \left( e_j, f_j \right) \right) .
    \end{align*}
    Then we can see that $\left| W\left( e_{a_k + 1}, \ldots, e_{b_k}; f_{a_k}^* \right) - N_k \left( e_{a_k + 1}, \ldots, e_{b_k}; f_{a_k}^* \right) \right| \leq p D \|P\|$, and\linebreak
    $N_k \left( e_{a_k + 1}, \ldots, e_{b_k}; f_{a_k}^* \right) = N_k \left( e_{a_k + 1}, \ldots, e_{b_k}; f_0 \right) ,$
    meaning that
    $$\left| W\left( e_{a_k}, \ldots, e_{b_k}; f_{a_k}^* \right) - W \left( e_{a_k + 1}, \ldots, e_{b_k}; f_0 \right) \right| \leq 2 p D \|P\| .$$
    Therefore
    \begin{align*}
    W \left( e_1, \ldots, e_{b_{k_0}}; f_0 \right)
    \geq & (k_0 - 1) p D \|P\| + 2 p k_0 D \|P\| + \sum_{k = 1}^{k_0} W \left( e_{a_k + 1}, \ldots, e_{b_k} ; f_0 \right) \\
    =    & (k_0 - 1) p D \|P\| + 2 p k_0 D \|P\| + \sum_{k = 1}^{k_0} \NAval \left( p k \mid e_0, f_0 \right) .
    \end{align*}
    
    It follows that $\lim_{n \to \infty} \frac{1}{n} W \left( e_1, \ldots, e_n ; f_0 \right) = \lambda$. This means that if $\cS_B$ is any strategy for Bob in the infinite-round non-alternating game, and $(e_j)_{j = 1}^\infty, (f_j)_{j = 1}^\infty$ are the walks produced by $\left( \cS_A^*, \cS_B \right)$, then $$\liminf_{n \to \infty} \frac{1}{n} \sum_{j = 1}^n P(e_j, f_j) \geq \liminf_{n \to \infty} \frac{1}{n} W \left( e_0, e_1, \ldots, e_n ; f_0 \right) = \lambda .$$

    Bob's strategy $\cS_B^*$ can be described as follows: If Alice chooses the walk $(e_j)_{j = 1}^\infty$, then Bob chooses a walk $(f_j)_{j = 0}^\infty$ such that
    \begin{align*}
    f_{a_k} & = f_0,    & W \left( e_{a_k}, \ldots, e_{b_k}; f_0 \right)  & = \sum_{j = a_k + 1}^{b_k} P(e_j, f_j)
    \end{align*}
    for all $k \in \N$. A similar computation to the one we just did will show that if $\cS_A$ is any strategy for Alice in the infinite-round non-alternating game, and $(e_j)_{j = 1}^\infty, (f_j)_{j = 1}^\infty$ are the walks produced by $\left( \cS_A, \cS_B^* \right)$, then
    $$\limsup_{n \to \infty} \frac{1}{n} \sum_{j = 1}^n P(e_j, f_j) \leq \limsup_{k \to \infty} \frac{1}{p k} \NAval(p k \mid e_0, f_0) = \lambda .$$
    
    Taken together, this implies that $\left( \cS_A^*, \cS_B^* \right)$ is a Nash equilibrium for the infinite-round non-alternating game, and thus that
    $\NAval(e_0, f_0) = \lambda .$
\end{proof}

\begin{example}\label{ex:non_alt_reducible_graphs_infinite_limit_counterexample}
There exist (reducible) graphs $G, H$, an integer-valued score function $P:\cE{G} \times \cE{H} \to \Z$,
    and initial edges $(e_0,f_0) \in \cE{G} \times \cE{H}$ so that  
    \[
    \NAval(e_0,f_0) \ne \lim_{n \to \infty} \frac{1}{n} \NAval( n \mid e_0,f_0).
    \]
This example appears in the proof of \cite[Proposition 1.11]{young2025adversarialergodicoptimization}, albeit in different language.
\begin{figure}
\begin{tikzpicture}[node distance = 10mm, main/.style = {draw, circle}] 
\node[main] (1) {$P$};
\node[main] (2) [below left of=1] {$X$}; 
\node[main] (3) [below right of=1] {$Y$};
\node[main] (4) [above right of=3] {$M$};
\node[main] (5) [below right of=4] {$Z$};

\Loop[dist=1cm,dir=NOEA,label=$e^-$,labelstyle=above right](4)  
\Loop[dist=1cm,dir=SOWE,label=$f^+$,labelstyle=below left](2)  
\Loop[dist=1cm,dir=NOWE,label=$e^+$,labelstyle=above left](1)  
\Loop[dist=1cm,dir=SOEA,label=$f^-$,labelstyle=below right](5)

\draw[->] (1) to node [above] {$e^0$} (4);
\draw[->] (3) to node [below] {$f^1$} (2);
\draw[->] (3) to node [below] {$f^2$} (5);
\end{tikzpicture}
\caption{Pictured here are the graphs referenced in \Cref{ex:non_alt_reducible_graphs_infinite_limit_counterexample}, where the top row depicts the graph $G$ with vertices $\cV{G} = \{P, M\}$, while the bottom row depicts the graph $H$ with vertices $\cV{H} = \{X,Y, Z\}$.}
\label{fig:graph_2}
\end{figure}
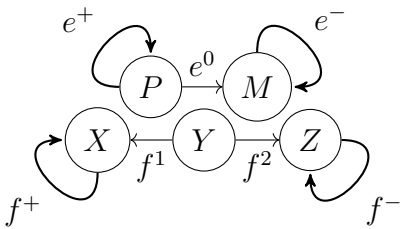
Let $G$ and $H$ be the graphs depicted in  \Cref{fig:graph_2}, and let $P:\cE{G} \times \cE{H} \to \R$ be the score function given by
\begin{align*}
P(e, f) = \begin{cases}
1   & \textrm{if $(e, f) \in \left\{ (e^+, f^+), (e^-, f^-) \right\}$,} \\
-1  & \textrm{if $(e, f) \in \left\{ (e^+, f^-), (e^-, f^+) \right\}$,} \\
0   & \textrm{otherwise.}
\end{cases}
\end{align*}
For simplicity, we present this example using initial vertices, rather than edges. Suppose that Alice starts at the vertex labeled $P$ and Bob starts at the vertex labeled $Y$.
In the $n$-round non-alternating game and in the infinite-round non-alternating game, Bob has only a single choice: He can use the edge $e^1$ or $e^2$ on the first round. After that, Bob is forced to keep looping using the same edge.
Alice's only choice is at what time $T \in \N$ to move from $P$ to $M$ (if ever).
In the $n$-round game, as Bob wants to minimize the total payoff, up to a constant, his best choice is to choose $f^+$ if $T$ is greater than $\frac{n}{2}$ and $f^-$ otherwise. Using this strategy, Bob can guarantee a total payoff of at most $1$ in a finite-round game. Likewise, in the $n$-round game, Alice can choose to use the loop $e^-$ for the first $n/2$ rounds and transition to $M$ via $e_0$.
Using this strategy, Alice can guarantee a total payoff of at least $-1$.
So $$\lim_{n \to \infty} \frac{1}{n} \NAval(n \mid P, Y) = 0.$$
On the other hand, in the infinite-round game, if Alice chooses to stay in $P$ forever, Bob can choose to transition directly to the vertex $Z$, and if Alice chooses to transition to $M$ eventually, Bob can choose to transition directly to the vertex $X$. Either way, the total payoff for the infinite-round game will be $\NAval(P, Y) = -1 .$
\end{example}

\section{Open questions about the non-alternating game}\label{section:open_questions}

While it is possible to compute the Nash equilibrium for an alternating game in pseudo-polynomial time (c.f. \cite[Theorem 2.3]{ZwickPaterson}), it is not so clear how to efficiently compute a Nash equilibrium for the $n$-round non-alternating game.

\begin{quest}\label{qn:compute_value_of_finite_game}
Is there an efficient algorithm to compute $\NAval(n \mid e_0,f_0)$?
\end{quest}

It is intuitively clear that the alternating game is more favorable to Alice than the non-alternating game. This is because in the non-alternating game, Alice has no information about Bob's moves, while Bob has perfect information about hers. In the alternating game, Alice and Bob take turns moving, meaning that Alice has the opportunity to respond to his plays dynamically. To put it more formally:

\begin{prop}\label{prop:comparing_alt_and_non_alt_games}
For every $n \in \N$ we have
    \begin{equation*}
        \NAval(n \mid e_0,f_0) \le \Aval(n \mid e_0,f_0).
    \end{equation*}
\end{prop}

\begin{example}\label{ex:alt_and_non_alt_games_different_limits}
The inequality in \Cref{prop:comparing_alt_and_non_alt_games} can be strict. In fact, it is possible to present examples where $\Aval(n \mid v, w) - \NAval(n \mid v, w) = \Omega(n)$. Consider the graphs $G, H$ given by
\begin{align*}
\cV{G}  & = \{W, X, Y, Z\} ,    & \cV{H}    & = \cV{G} , \\
\cE{G}  & = \cV{G} \times \cV{G} ,  & \cE{H}    & = \cE{G} \setminus \{ (W, Y), (Y, W), (X, Z), (Z, X) \} .
\end{align*}
This is illustrated in \Cref{figure:bob_chasing_alice},
where the solid arrows denote edges in both graphs $G, H$, dashed arrows denote edges present in $G$ but not in $H$, and double-headed arrows denote one edge in each direction.
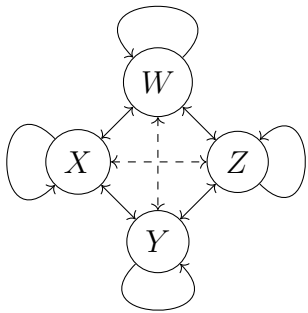
\begin{figure}
\begin{tikzpicture}[node distance={15mm}, main/.style = {draw, circle}]
	\node[main] (1) {$W$};
	\node[main] (2) [below left of =1] {$X$};
	\node[main] (3) [below right of =1] {$Z$};
	\node[main] (4) [below left of =3] {$Y$};
	\draw[<->] (1) to [out=225, in=45] node[midway, above left] {} (2);
	\draw[<->] (1) to [out=315, in=135] node[midway, below left] {} (3);
	\draw[<->] (3) to [out=225, in=45] node[midway, above left] {} (4);
	\draw[<->] (2) to [out=315, in=135] node[midway, below left] {} (4);
        \draw[<->, dashed] (1) to [out = 270, in=90] (4);
        \draw[<->, dashed] (2) to [out = 0, in=180] (3);
	\draw[->] (1) to [out=45, in=135, looseness=5] node[midway, above] {} (1);
	\draw[->] (2) to [out=135, in=225, looseness=5] node[midway, left] {} (2);
	\draw[->] (3) to [out=315, in=45, looseness=5] node[midway, right] {} (3);
	\draw[->] (4) to [out=225, in=315, looseness=5] node[midway, below] {} (4);
\end{tikzpicture}
\caption{The graphs $G, H$ referenced in \Cref{ex:alt_and_non_alt_games_different_limits}. Solid arrows denote edges present in both $G$ and $H$. Dashed arrows denote edges of $G$ that are \emph{not} edges of $H$.}
\label{figure:bob_chasing_alice}
\end{figure}
Let $P : \cE{G} \times \cE{H} \to \{-1, 1\}$ be the score function
\begin{align*}
P(e, f) & = \begin{cases}
1   & \textrm{if $t(e) \neq t(f)$,} \\
-1  & \textrm{if $t(e) = t(f)$.}
\end{cases}
\end{align*}
We can think of this as the game where Bob ``scores" in turn $j$ if the edge he chooses in turn $j$ has the same target vertex as the edge Alice chose earlier that turn, and Alice scoring in turn $j$ otherwise. In the $n$-round alternating game, Alice has a strategy which ensures that she scores in every round. For every $v \in \cV{H} = \cV{G}$, there exists a unique $\theta(v) \in \cV{H}$ such that $(v, \theta(v)) \in \cE{G} \setminus \cE{H}$, corresponding to the opposite vertex as depicted in \Cref{figure:bob_chasing_alice}. If it's turn $j$, Alice's last move was $e_{j - 1}$, and Bob's last move was $f_{j - 1}$, then Alice can play $e_j = (t(e_{j - 1}), \theta(t(f_{j - 1}))) \in \cE{G}$, meaning that for every $f_j \in N_H(f_{j-1})$, she has that $t(e_j) \neq t(f_j)$. This means that in every turn, Alice can prevent Bob from scoring, thus ensuring that $S_A(n \mid e_0, f_0) = n$ for all $e_0 \in \cE{G}, f_0 \in \cE{H}$. On the other hand, we can consider the non-alternating game with initial edges $e_0 \in \cE{G}, f_0 \in \cE{H}$, where Alice plays some walk $(e_j)_{j = 0}^n$. Then Bob can choose a walk $(f_j)_{j = 0}^n$ such that $t(f_{2 k}) = t(e_{2 k})$ for all $k = 1, \ldots, \lfloor n / 2 \rfloor$, so he scores approximately every other round, meaning that $\NAval(n \mid e_0, f_0) \leq 2$. Thus $\Aval(n \mid e_0, f_0) - \NAval(n \mid e_0, f_0) \geq n - 2$ for all $e_0 \in \cE{G}, f_0 \in \cE{H}, n \in \N$. It follows from this that $\Aval(e_0, f_0) > \NAval(e_0, f_0)$.
\end{example}

\begin{quest}
Consider directed graphs $G, H$ such that every vertex is the source of at least one edge. Does the limit  $\lim_{n \to \infty} \frac{1}{n} \NAval(n \mid e, f)$ exist for all $(e,f )\in \cE{G}\times\cE{H}$?
\end{quest}

\begin{quest}
Is there an efficient algorithm to compute $\NAval(e, f)$?
\end{quest}

It is possible with a more careful analysis to show when $G, H$ are irreducible that \linebreak$\left| \NAval(e, f) - \frac{1}{n} \NAval(n \mid e, f) \right| = O \left( 1 / \sqrt{n} \right)$. Given that we currently have no efficient means by which to calculate $\NAval(n \mid e, f)$ (c.f. \Cref{qn:compute_value_of_finite_game}), this estimate suggests that attempting to estimate $\NAval(e, f)$ by $\frac{1}{n} \NAval(n \mid e, f)$ will be very inefficient. We also note that in the case where $G, H$ are primitive, there are several equivalent formulations of $\NAval(\cdot, \cdot)$ in terms of the adversarial ergodic optimization, c.f. \cite[Corollary 2.14]{young2025adversarialergodicoptimization}.

\begin{quest}\label{qn:periodic_non_alt}
Consider irreducible graphs $G, H$. Does there exist always exist a Nash equilibrium 
for the non-alternating infinite-round game that always outputs an eventually periodic walk in $G \times H$? 
\end{quest}

The answer to \Cref{qn:periodic_non_alt} is negative if we remove the hypothesis that both graphs are irreducible. We provide two examples, illustrating slightly different ``obstructions to periodicity."

\begin{example}\label{ex:irrational_example}
Like in \Cref{ex:non_alt_reducible_graphs_infinite_limit_counterexample}, we will take our ``initial conditions" to be a pair of vertices rather than edges in order to streamline our presentation. Consider the graphs $G, H$ depicted in \Cref{fig:no_periodic_winning_strat_in_non_alt_irr}.
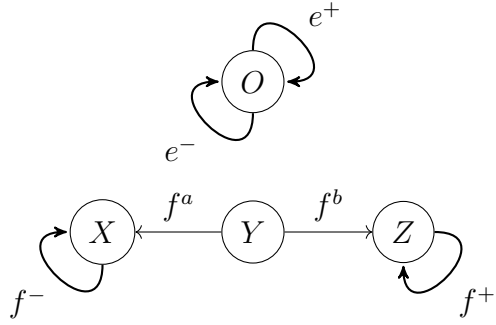
\begin{figure}
\begin{tikzpicture}[node distance = 20mm, main/.style = {draw, circle}] 
\node[main] (1) {$O$};
\node[main] (2) [below of=1] {$Y$}; 
\node[main] (3) [left of=2] {$X$};
\node[main] (4) [right of=2] {$Z$};

\Loop[dist=1cm,dir=NOEA,label=$e^+$,labelstyle=above right](1)  
\Loop[dist=1cm,dir=SOWE,label=$e^-$,labelstyle=below left](1)
\Loop[dist=1cm,dir=SOWE, label=$f^-$,labelstyle=below left](3)
\Loop[dist=1cm,dir=SOEA, label=$f^+$, labelstyle=below right](4)

\draw[->] (2) to node [above] {$f^a$} (3);
\draw[->] (2) to node [above] {$f^b$} (4);
\end{tikzpicture}
\caption{Pictured here are the graphs $G, H$ referenced in \Cref{ex:irrational_example}, where the top row depicts a graph $G$ with vertices $\cV{G} = \{O\}$, and the bottom row depicts a graph $H$ with vertices $\cV{H} = \{X, Y, Z\}$.}
\label{fig:no_periodic_winning_strat_in_non_alt_irr}
\end{figure}
We choose two rationally independent positive real numbers $\alpha,\beta>0$ (e.g. $\alpha=1,\beta=\sqrt{2}$), and define the score function $P : \cE{G} \times \cE{H} \to \R$ as follows:
\[
P(e,f) = \begin{cases}
    \alpha & \textrm{if } (e, f )= \left(e^+,f^+\right), \\
    -\alpha & \textrm{if } (e,f)=\left(e^+,f^-\right), \\
    \beta & \textrm{if } (e,f)=\left(e^-,f^-\right), \\
    -\beta & \textrm{if } (e,f)=\left(e^-,f^+\right), \\
    0 & \textrm{otherwise.}
  \end{cases}
\]

Starting from the initial vertex $Y$, Bob can only choose between two infinite walks \linebreak$\left(f^a,f^-, f^-, \ldots\right), \left(f^b,f^+, f^+, \ldots\right)$.
If Alice's strategy is to choose an eventually periodic walk $(e_j)_{j = 1}^\infty$, then the following limits exist: 
\begin{equation}\label{eq:p_plus_minus}
p_+:= \lim_{n \to \infty}\frac{1}{n}\#\{1\le j \le n~:~ e_j = e^{+}\} \mbox{ and } p_- := \lim_{n \to \infty}\frac{1}{n}\#\{1\le j \le n~:~ e_j = e^{-}\}.    
\end{equation}
If the sequence $(e_j)_{j = 1}^\infty$ is eventually periodic, then $p_+,p_-$ are rational numbers such that $p_+ + p_-=1$. 
In this case, the payoff is $\pm \left( \alpha p_+ - \beta p_- \right)$,
depending on Bob's choice of walk.
So Bob can secure a payoff of at most
\[- |\alpha p_+ - \beta p_- |.\]
Since $p_+, p_-$ are rational and sum to $1$, the linear independence of $\alpha,\beta$ over $\mathbb{Q}$ implies that $\alpha p_+ - \beta p_- \ne 0$, so whenever Alice's strategy chooses an eventually periodic output, Bob can secure a strictly negative payoff.
On the other hand, Alice can choose a non-eventually periodic walk $(e_j)_{j = 1}^\infty$ such that
the limits in \eqref{eq:p_plus_minus} exist, and $\alpha p_+ = \beta p_-$, in which case the payoff will be $0$, regardless of Bob's strategy. This proves that there is no eventually periodic Nash equilibrium.
\end{example}

\begin{example}\label{ex:integer_valued_counterexample}
\begin{figure}
\begin{tikzpicture}[node distance = 20mm, main/.style = {draw, circle}] 
\node[main] (1) {$M$};
\node[main] (2) [below of=1] {$X$}; 
\node[main] (3) [right of=2] {$Y$};
\node[main] (4) [right of=3] {$Z$};
\node[main] (5) [above of=4] {$P$};

\Loop[dist=1cm,dir=NOWE,label=$e^{-}$,labelstyle=above left](1)  
\Loop[dist=1cm,dir=NOEA,label=$e^{+}$,labelstyle=above right](5)
\Loop[dist=1cm,dir=SOWE, label=$f^{-}$,labelstyle=below left](2)
\Loop[dist=1cm,dir=SOEA, label=$f^+$, labelstyle=below right](4)

\draw[->] (3) to node [above] {$f^a$} (2);
\draw[->] (3) to node [above] {$f^b$} (4);
\draw[->] (1) to [out=15, in=165] node [above] {$e^a$} (5);
\draw[->] (5) to [out=195, in=345] node [below] {$e^b$} (1);
\end{tikzpicture}
\caption{Pictured here are the graphs $G, H$ referenced in \Cref{ex:irrational_example}, where the top row depicts the graph $G$ with vertices $\cV{G} = \{P, M\}$, and the bottom row depicts a graph $H$ with vertice $\cV{H} = \{X, Y, Z\}$.}
\label{fig:no_periodic_winning_strat_in_non_alt_but_rational}
\end{figure}
Let $G, H$ be the graphs depicted in \Cref{fig:no_periodic_winning_strat_in_non_alt_but_rational}, and let  $P : \cE{G} \times \cE{H} \to \Z$ to be the function
\[
P(e, f) = \begin{cases}
1 & \textrm{if } (e,f) \in \left\{ \left(e^+,f^+\right), \left(e^-,f^-\right) \right\}, \\
-1  & \textrm{otherwise.}
\end{cases}
\]
As in \Cref{ex:irrational_example}, if Bob starts from the initial vertex $Y$, then he can only choose between the infinite walks $\left( f^a,f^-,f^-,\ldots\right), \left( f^b,f^+,f^+,\ldots\right)$.
If Alice's strategy is to choose an eventually periodic walk $(e_j)_{j = 1}^\infty$, then the limits $p_+, p_-$ exist, as does the limit $p_0 : = \lim_{n \to \infty} \frac{1}{n} \# \left\{ 1 \leq j \leq n ~ : e_j \not \in \left\{ e^+, e^- \right\} \right\}$.
It also follows that $p_+,p_-,p_0 \geq 0$ and $p_+ + p_- + p_0 = 1$.
Furthermore, unless $e_j \in \left\{ e^+, e^- \right\}$ for all sufficiently large $n \in \N$, we have that $p_0>0$, because every period must contain a ``transition edge" between $e^+$ and $e^-$.
In either case, Bob can secure a payoff of
\[
-p_0 -|p_+-p_-|,
\]
which is strictly negative.
On the other hand, Alice can choose a (non-eventually periodic) walk $(e_j)_{j = 1}^\infty$ such that $p_+=p_-=\frac{1}{2}$ and $p_0=0$, e.g.
\[
\left( e^-,e^a,e^+,e^b,e^-,e^-,e^a,e^+,e^+,e^b,\ldots,\underbrace{e^-,\ldots,e^-}_n,e^a,\underbrace{e^+,\ldots,e^+}_n,e^b,\ldots \right).
\]
In this case, the payoff will be $0$, regardless of Bob's strategy, again proving that there is no eventually periodic Nash equilibrium in this example.
\end{example}

\bibliographystyle{amsplain}
\bibliography{lib}
\end{document}